\documentclass[letterpaper]{article}

\usepackage[margin=1in]{geometry}
\usepackage[utf8]{inputenc}
\usepackage{amsmath,amssymb,mathtools}
\usepackage{amsthm}
\usepackage{thmtools}
\usepackage{microtype}
\usepackage{booktabs}
\usepackage{array}
\usepackage{placeins}
\usepackage{natbib}
\usepackage[colorlinks=true,allcolors=blue]{hyperref}
\usepackage[capitalize,noabbrev,nameinlink]{cleveref}
\usepackage{xcolor}

\numberwithin{equation}{section}

\declaretheorem[numberlike=equation]{theorem}
\declaretheorem[numberlike=equation]{lemma}
\declaretheorem[numberlike=equation]{corollary}
\declaretheorem[numberlike=equation]{proposition}
\declaretheorem[numberlike=equation]{conjecture}

\declaretheoremstyle[spaceabove=6pt,spacebelow=6pt,bodyfont=\normalfont]{defstyle}
\declaretheorem[numberlike=equation,style=defstyle]{definition}
\declaretheorem[numberlike=equation,style=defstyle]{example}
\newenvironment{algorithmblock}[1]
  {\par\smallskip\noindent\textbf{#1.}\ }
  {\par\smallskip}

\crefname{theorem}{Theorem}{Theorems}
\crefname{lemma}{Lemma}{Lemmas}
\crefname{corollary}{Corollary}{Corollaries}
\crefname{proposition}{Proposition}{Propositions}
\crefname{conjecture}{Conjecture}{Conjectures}
\crefname{definition}{Definition}{Definitions}
\crefname{example}{Example}{Examples}

\DeclarePairedDelimiter{\card}{\lvert}{\rvert}
\newcommand{\E}{\mathbb{E}}
\newcommand{\Prb}{\mathbb{P}}
\newcommand{\PJR}{\mathrm{PJR{+}}}
\newcommand{\EJR}{\mathrm{EJR{+}}}

\title{Learning Proportional Committees from Violation Feedback}
\author{Frank Connor\\
Massachusetts Institute of Technology\thanks{The author thanks Bailey Flanigan for valuable discussions and detailed comments that substantially improved the presentation of this work.}}
\date{}

\begin{document}
\maketitle
\begin{abstract}
We study violation-feedback learning of proportionally representative
approval-based committees. In each round, a learner proposes a committee
of size $k$. An oracle either accepts the proposal or adversarially selects a
representation violation with respect to a single fixed hidden approval
profile. We compare \emph{full-witness feedback}, which reveals the violation
level, an omitted candidate, and the affected voter group, with
\emph{candidate-only feedback}, which reveals only that candidate. The target
notions are proportional justified representation plus (PJR+) and extended
justified representation plus (EJR+).

In every setting we study, the number of rejected proposals can be bounded
solely in terms of $k$, with no dependence on the numbers of voters and
candidates. For PJR+, the optimal deterministic and randomized rejection
complexities equal $k$ under both feedback models. For EJR+, the picture is more nuanced. Under full-witness feedback, we prove
an $\Omega(k^{3/2})$ deterministic lower bound and give a deterministic
polynomial-time algorithm using $O(k^2\log k)$ rejections. Under candidate-only
feedback, randomization achieves $O(k^2\log k)$ expected rejections via
uniform random deletion, while deterministic exhaustive branching gives a
$2^{O(k^2(\log k)^2)}$ rejection bound. Even with full-witness feedback,
randomized learners may require $k$ rejections.
\end{abstract}

\section{Introduction}

Approval-based committee selection is usually a one-shot process: every voter
submits an approval ballot, and an aggregation rule selects a committee from
the resulting profile. Forming and reporting
these approval sets up front may be costly, particularly when voters have
limited information, incomplete preferences, or little incentive to evaluate
candidates who may prove irrelevant to the eventual outcome. Recent work on
preference elicitation and incomplete approval information reduces this burden
through selective queries, assumptions on the structure of missing
preferences, and sampled approval ballots; see, e.g.,
\citep{ConitzerSandholm2002,BoutilierRosenschein2016,HalpernEtAl2026,%
SpringhamEtAl2025,SpringhamEtAl2026,Kehne2026}.

We introduce an orthogonal approach: \emph{violation-feedback learning}. The learner cannot query individual approvals. Instead, it repeatedly proposes a full committee and receives feedback
when the proposal fails proportional representation. A sufficiently large
underrepresented group may then identify an omitted candidate whom its members
approve, prompting the learner to revise the committee. Thus, the learner need
not receive the complete profile up front: feedback about candidate support is
elicited only when a candidate becomes relevant to a concrete representation
failure.

This interaction belongs to the trial-and-error framework of \citet{BeiChenZhang2013}, in which an algorithm proposes solutions to a
hidden instance and learns from violated constraints. \citet{LiuKempeMicha2026} recently applied this framework to fair
division. We ask the analogous question for proportional representation: can
social-choice guarantees be reached from violations of proposed outcomes
without first revealing the full underlying instance?

We formalize this interaction as follows. The learner knows the voter set
$N$, the candidate set $C$, their sizes $n$ and $m$, the target committee size
$k$, and the target proportionality axiom, but not the fixed approval profile.
We focus on \emph{proportional justified representation plus} (PJR+) and
\emph{extended justified representation plus} (EJR+). Suppose a group large
enough to claim $\ell$ seats unanimously approves an omitted candidate. PJR+
requires the group to collectively approve at least $\ell$ committee members,
whereas EJR+ requires some member of the group to personally approve at least
$\ell$ committee members. EJR+ implies PJR+, while PJR+ and EJR+ strengthen
proportional justified representation (PJR) and extended justified
representation (EJR), respectively. Although they provide stronger guarantees,
for the plus variants, a violation can be found in polynomial time whenever
one exists, whereas verifying ordinary PJR or EJR is coNP-complete
\citep{AzizEtAl2017,AzizEtAl2018,BrillPeters2023}. Moreover, every profile
admits an EJR+ committee~\citep{BrillPeters2023}.

We model this objection process through a truthful oracle, which accepts a
committee satisfying the target axiom and otherwise returns an adversarially
chosen violation for the fixed profile. A violation consists of a claimed
number of seats $\ell$, an omitted candidate $c$, and an underrepresented
witness group $S$. Under
\emph{full-witness feedback}, the learner observes the entire triple
$(\ell,c,S)$; under \emph{candidate-only feedback}, it observes only $c$,
hiding both the affected voters and the size of their claim. Full-witness
feedback can model a centralized auditor with access to the hidden ballots and
also serves as a strong-information benchmark. Candidate-only feedback captures
a privacy-limited objection process in which only the commonly approved omitted
candidate is disclosed.

We measure the interaction by its \emph{rejection complexity}: the worst-case
number of rejected proposals before the first acceptance. Our central question
is:

\begin{center}
\emph{How many rejected committees are necessary and sufficient to find a
proportionally representative committee when the approval profile is hidden?}
\end{center}

A priori, the number of rejections could grow with the electorate or candidate
set: larger elections contain more possible underrepresented groups and
omitted candidates, while each rejection identifies only one local failure.
Such a dependence would make the mechanism increasingly costly precisely in
the large, low-information elections that motivate it. Nevertheless, we find surprising rejection bounds that depend only on $k$, not on $n$ or
$m$.

\subsection{Our results}

For an axiom $\mathsf X$, let $\rho=\mathrm{full}$ and
$\rho=\mathrm{cand}$ denote full-witness and candidate-only feedback,
respectively. We write $Q_\rho^{\mathsf X}(k)$ and
$\widetilde Q_\rho^{\mathsf X}(k)$ for the optimal deterministic and
randomized rejection complexities. Formal definitions appear in the next section.

Our upper bounds follow a progression in the form of progress that the learner
can certify. For PJR+, returned candidates can be permanently retained; for
EJR+ with full witnesses, this is replaced by observable progress in a PAV
potential; with candidate-only feedback, only average potential progress
remains.

We first resolve PJR+ exactly. A learner can retain every candidate returned
by the oracle in all later proposals. A packing argument shows that at most $k$ such candidates can be
returned, while a matching lower bound holds even for randomized learners
receiving full-witness feedback. Thus,
$$
Q_{\mathrm{full}}^{\PJR}(k)
=
Q_{\mathrm{cand}}^{\PJR}(k)
=
\widetilde Q_{\mathrm{full}}^{\PJR}(k)
=
\widetilde Q_{\mathrm{cand}}^{\PJR}(k)
=
k.
$$

We next turn to EJR+ under full-witness feedback. Permanent retention need
not remain feasible, but each full-witness response reveals enough approval
information to certify a PAV-improving swap; sufficiently strong PAV local
optimality, in turn, guarantees EJR+. This yields a deterministic
polynomial-time algorithm with at most $k^2H_k=O(k^2\log k)$ rejections. A
complementary construction shows that deterministic EJR+ learning may require
$\Omega(k^{3/2})$ rejections even with full-witness feedback.

We then study EJR+ under candidate-only feedback. The returned candidate still guarantees a PAV improvement on average
over the $k$ possible deletions, but the learner does not observe the witness
group needed to identify a good deletion. Deterministically, exhaustive branching gives
$Q_{\mathrm{cand}}^{\EJR}(k)\le 2^{O(k^2(\log k)^2)}$, albeit
inefficiently. Randomization exploits the same average guarantee directly:
uniform random deletion yields
$\widetilde Q_{\mathrm{cand}}^{\EJR}(k)\le k^2H_k=O(k^2\log k)$.
We also prove a randomized lower bound of $k$, even with full-witness feedback.
Finally, every learner that commits to a single insertion--deletion path may
require $\Omega(k^{3/2})$ rejections.

For PJR+, the strength of the feedback makes no difference: all four
complexities equal $k$. For EJR+, by contrast, the known deterministic upper
bounds differ substantially between the two feedback models, although it
remains open whether this reflects a genuine separation. Table~\ref{tab:results}
summarizes our bounds. An entry indexed by $\rho$ applies to both feedback
models, $\rho\in\{\mathrm{full},\mathrm{cand}\}$.

\begin{table*}[t]
\centering
\footnotesize
\renewcommand{\arraystretch}{1.25}
\setlength{\tabcolsep}{4pt}
\begin{tabular}{@{}
  >{\raggedright\arraybackslash}p{0.08\linewidth}
  >{\raggedright\arraybackslash}p{0.22\linewidth}
  >{\raggedright\arraybackslash}p{0.31\linewidth}
  >{\raggedright\arraybackslash}p{0.31\linewidth}
@{}}
\toprule
\textbf{Axiom}
&
\textbf{Complexity}
&
\textbf{Lower bound}
&
\textbf{Upper bound}
\\
\midrule

PJR+
&
$Q_{\rho}^{\PJR}(k)$
&
$k$ {\scriptsize(Thm.~\ref{thm:pjr-plus-lower-bound})}
&
$k$ {\scriptsize(Thm.~\ref{thm:pjr-plus-upper-bound})}
\\
\addlinespace[3pt]

&
$\widetilde Q_{\rho}^{\PJR}(k)$
&
$k$ {\scriptsize(Thm.~\ref{thm:randomized-full-witness-linear-lower-bound})}
&
$k$ {\scriptsize(Thm.~\ref{thm:pjr-plus-upper-bound})}
\\

\midrule

EJR+
&
$Q_{\mathrm{full}}^{\EJR}(k)$
&
$\Omega(k^{3/2})$
{\scriptsize(Cor.~\ref{cor:asymptotic-full-witness-ejr-plus-lower-bound})}
&
$O(k^2\log k)$
{\scriptsize(Thm.~\ref{thm:full-witness-ejr-plus-upper-bound})}
\\
\addlinespace[3pt]

&
$Q_{\mathrm{cand}}^{\EJR}(k)$
&
$\Omega(k^{3/2})$
{\scriptsize(Cor.~\ref{cor:asymptotic-full-witness-ejr-plus-lower-bound})}
&
$2^{O(k^2(\log k)^2)}$
{\scriptsize(Thm.~\ref{thm:deterministic-candidate-only-ejr-plus-upper-bound})}
\\
\addlinespace[3pt]

&
$\widetilde Q_{\rho}^{\EJR}(k)$
&
$k$ {\scriptsize(Thm.~\ref{thm:randomized-full-witness-linear-lower-bound})}
&
$O(k^2\log k)$
{\scriptsize(Thm.~\ref{thm:randomized-candidate-only-ejr-plus-upper-bound})}
\\

\bottomrule
\end{tabular}
\caption{Rejection-complexity bounds.}
\label{tab:results}
\end{table*}

\subsection{Related work}

The closest application of violation feedback is the fair-allocation model of
\citet{LiuKempeMicha2026}. \citet{MichaVarsamis2025} study voting from local
improvement feedback to proposed outcomes, with the goal of computing a
prescribed voting rule. Our oracle instead returns a group certificate of a
proportionality violation, and the learner seeks any committee satisfying the
axiom. More broadly, our model belongs to the trial-and-error framework of
\citet{BeiChenZhang2013}. Unlike shortest-path interactive learning
\citep{EmamjomehZadehKempe2017} and equivalence-query learning
\citep{Angluin1988}, however, a violation need not point toward a fixed target.

Complementary work reduces information requirements through incomplete,
structured, adaptively elicited, or sampled approval data
\citep{ConitzerSandholm2002,BoutilierRosenschein2016,HalpernEtAl2026,%
SpringhamEtAl2025,SpringhamEtAl2026,Kehne2026}. Our learner instead receives
feedback only on complete proposed committees. 

The proportionality axioms come from approval-based committee voting; see
\citet{LacknerSkowron2023}. \citet{AzizEtAl2017} introduced justified
representation and EJR, \citet{SanchezFernandezEtAl2017} introduced PJR, and
\citet{BrillPeters2023} formulated PJR+ and introduced EJR+. For our upper
bounds, we build on thresholded PAV local search
\citep{AzizEtAl2018,CaseyElkind2025}. Whereas arbitrarily small PAV
improvements may require much longer searches \citep{KraiczyElkind2024},
every improvement certified here is at least $n/k^2$. Under candidate-only
feedback, these gains are hidden, so our deterministic learner retains every
deletion branch instead.

\section{The Violation-Feedback Model}
\label{sec:model}

\subsection{Proportional representation}
\label{subsec:proportional-representation}

For a positive integer $r$, let $[r]:=\{1,\ldots,r\}$. For a finite set
$X$, let $\binom{X}{k}$ denote the family of its $k$-element subsets.
An \emph{approval-based committee election} is a tuple $(N,C,A,k)$,
where $N=[n]$ is the set of voters, $C=[m]$ is the set of candidates,
$A=(A_i)_{i\in N}$ is an approval profile with $A_i\subseteq C$ for
each $i\in N$, and $k\in[m]$ is the target committee size. A proposed
\emph{feasible committee} is a set $W\in\binom{C}{k}$.

We now formalize the certificates returned when a committee fails PJR+ or
EJR+. The definitions encode, respectively, collective and voter-specific
shortfalls in representation.

\begin{definition}[PJR+ violation~\cite{AzizLee2021,BrillPeters2023}]
\label{def:pjr-plus}
A \emph{PJR+ violation} of a feasible committee $W$ is a triple
$(\ell,c,S)$, where $\ell\in[k]$, $c\in C\setminus W$, and
$S\subseteq N$, such that
$$
|S|\ge \ell\frac{n}{k},
\qquad
c\in\bigcap_{i\in S}A_i,
\qquad
\left|W\cap\bigcup_{i\in S}A_i\right|<\ell.
$$
\end{definition}

In words, $S$ is large enough to claim $\ell$ representatives and unanimously
approves the omitted candidate $c$, but collectively approves fewer than
$\ell$ members of $W$.

\begin{definition}[EJR+ violation~\cite{BrillPeters2023}]
\label{def:ejr-plus}
An \emph{EJR+ violation} of a feasible committee $W$ is a triple
$(\ell,c,S)$, where $\ell\in[k]$, $c\in C\setminus W$, and
$S\subseteq N$, such that
\[
\begin{aligned}
|S|&\ge \ell\frac{n}{k},
&c&\in\bigcap_{i\in S}A_i,\\
|A_i\cap W|&<\ell
&&\text{for every }i\in S.
\end{aligned}
\]
\end{definition}

Here the deficiency is individual: every voter in $S$ approves fewer than
$\ell$ members of $W$. A committee satisfies PJR+ or EJR+ exactly when it
admits no violation of the corresponding type. In either type of violation,
$\ell$ is the \emph{level}, $c$ is the \emph{returned candidate}, and $S$ is
the \emph{witness group}. For example, a response $(2,c,S)$ says that $S$ is
large enough to claim two seats, all its members approve the omitted candidate
$c$, and the relevant representation requirement fails at level two. The
``plus'' certificates require the voters to agree only on the single omitted
candidate $c$, rather than on $\ell$ common candidates. EJR+ implies both EJR and PJR+, and PJR+ implies PJR; however,
EJR and PJR+ are incomparable~\cite{BrillPeters2023}.

\subsection{Feedback model and rejection complexity}
\label{subsec:feedback-complexity}

Fix a target axiom $\mathsf X\in\{\PJR,\EJR\}$. For a profile $A$ and a
feasible committee $W$, let
$$
\mathcal V_{\mathsf X}(A,W)
\subseteq [k]\times C\times 2^N
$$
be the set of $\mathsf X$-violations against $W$ under profile $A$. Thus,
$W$ satisfies $\mathsf X$ exactly when
$\mathcal V_{\mathsf X}(A,W)=\varnothing$.

One round proceeds as follows. Based on the proposal--response pairs observed
so far, the learner proposes a committee $W$. The oracle either accepts $W$ or
selects a violation $(\ell,c,S)$. After a rejection, the learner observes the
entire triple under full-witness feedback and only $c$ under candidate-only
feedback. The sequence of past proposal--response pairs is the interaction
history.

For fixed $n,m$, let $\mathcal A_{n,m}:=(2^C)^N$ be the set of approval
profiles, and let $\Delta(Y)$ denote the distributions over a finite set $Y$.
The corresponding non-accepting response alphabets are
$$
\mathcal R_{\mathrm{full}}=[k]\times C\times2^N,
\qquad
\mathcal R_{\mathrm{cand}}=C.
$$

For a feedback model $\rho\in\{\mathrm{full},\mathrm{cand}\}$, the
set of interaction histories is
$$
\mathcal H_\rho
:=
\bigcup_{t\ge0}
\left(\binom{C}{k}\times\mathcal R_\rho\right)^t.
$$

A deterministic learner under feedback model $\rho$ maps each history to its
next proposed committee; let
$\mathfrak L_{\rho}^{\mathrm{det}}(n,m,k)$ be the set of all maps
$\mathcal H_\rho\to\binom{C}{k}$. A randomized learner maps each history to a
distribution over the next committee and then draws using private randomness;
let $\mathfrak L_{\rho}^{\mathrm{rand}}(n,m,k)$ be the set of all maps
$\mathcal H_\rho\to\Delta(\binom{C}{k})$.

The profile $A$ is fixed throughout the interaction but hidden from the
learner. For a proposed committee $W$, the available full-witness responses
are the elements of $\mathcal V_{\mathsf X}(A,W)$, while a candidate-only
response $c$ is available when $(\ell,c,S)\in\mathcal V_{\mathsf X}(A,W)$
for some $\ell$ and $S$. An oracle strategy is a map
$$
\mathcal O:
\mathcal H_\rho\times\binom{C}{k}
\longrightarrow
\{\textsc{accept}\}\cup\mathcal R_\rho.
$$
It is \emph{truthful for $A$} if it returns \textsc{accept} exactly when
$\mathcal V_{\mathsf X}(A,W)=\varnothing$ and otherwise returns an available
response for $W$. Its dependence on the history allows the oracle to choose
adaptively and adversarially among the responses available for the current
committee.

For $\rho\in\{\mathrm{full},\mathrm{cand}\}$, let
$\mathfrak O_{\rho}^{\mathsf X}(A)$ denote the set of oracle strategies for
$A$ under feedback model $\rho$.

For a learner $\mathcal L$, profile $A$, and oracle
$\mathcal O\in\mathfrak O_{\rho}^{\mathsf X}(A)$, let
$\tau(\mathcal L,A,\mathcal O)$ be the number of rejected proposals before the
first acceptance, with $\tau(\mathcal L,A,\mathcal O)=\infty$ if no proposal
is accepted. When $\mathcal L$ is randomized, $\tau$ is a random variable.

\begin{definition}[Deterministic rejection complexity]
\label{def:deterministic-rejection-complexity}
The \emph{deterministic rejection complexity} of $\mathsf X$ under feedback
model $\rho$ is
$$
Q_{\rho}^{\mathsf X}(k)
:=
\sup_{\substack{n\ge1\\ m\ge k}}
\inf_{\mathcal L\in\mathfrak L_{\rho}^{\mathrm{det}}(n,m,k)}
\sup_{A\in\mathcal A_{n,m}}
\sup_{\mathcal O\in\mathfrak O_{\rho}^{\mathsf X}(A)}
\tau(\mathcal L,A,\mathcal O).
$$
Thus, for each $n,m$, the learner may depend on $N$, $C$, $k$, and
$\mathsf X$, but it must succeed against every hidden profile and every
oracle strategy for that profile.
\end{definition}

For a randomized learner, the oracle observes the realized history and the
current proposal before choosing its response, but not the learner's future
random choices.

\begin{definition}[Randomized rejection complexity]
\label{def:randomized-rejection-complexity}
The \emph{randomized rejection complexity} of $\mathsf X$ under feedback
model $\rho$ is
\[
\begin{aligned}
\widetilde Q_{\rho}^{\mathsf X}(k)
:=\sup_{\substack{n\ge1\\ m\ge k}}
&\inf_{\mathcal L\in\mathfrak L_{\rho}^{\mathrm{rand}}(n,m,k)}
\sup_{A\in\mathcal A_{n,m}}\\[-2pt]
&\sup_{\mathcal O\in\mathfrak O_{\rho}^{\mathsf X}(A)}
\E\!\left[\tau(\mathcal L,A,\mathcal O)\right],
\end{aligned}
\]
where the expectation is over the learner's private randomness.
\end{definition}

These quantities count rejected proposals; including the final accepted
committee increases the total number of proposals by one. We impose no computational restriction on the learner.

Randomization cannot hurt, and full-witness feedback is at least as
informative as candidate-only feedback. Therefore,
\begin{equation}
\begin{aligned}
\widetilde Q_{\rho}^{\mathsf X}(k)
&\le Q_{\rho}^{\mathsf X}(k)
&&\text{for every }\rho\in\{\mathrm{full},\mathrm{cand}\},\\
Q_{\mathrm{full}}^{\mathsf X}(k)
&\le Q_{\mathrm{cand}}^{\mathsf X}(k),\\
\widetilde Q_{\mathrm{full}}^{\mathsf X}(k)
&\le \widetilde Q_{\mathrm{cand}}^{\mathsf X}(k).
\end{aligned}
\label{eq:complexity-monotonicity}
\end{equation}

A lower-bound construction ultimately consists of one fixed profile and one
truthful oracle strategy. For deterministic learners, constructing responses
online is only a proof device: the completed interaction history must be
realizable by one fixed profile. For randomized learners, the profile may
depend on the learner but must be fixed independently of the learner's
realized random choices.

\section{Retained-Set Progress}
\label{sec:pjr-plus-complexity}

The collective nature of PJR+ supports a simple retention rule. If a later
witness group intersects an earlier one, it collectively approves the
candidate returned for the earlier group. This intersection property yields
the packing condition used below.

\begin{algorithmblock}{Retention Algorithm}
Maintain the set $L$ of candidates returned so far, initially
$L=\varnothing$. Propose any feasible committee $W$ containing $L$; after a
rejection returning $c$, set $L\leftarrow L\cup\{c\}$ and repeat.
\end{algorithmblock}
Here ``retain'' means that every subsequent proposal contains the candidate.
Unlike profile-based greedy constructions of proportional committees
(e.g.,~\cite{BrillPeters2023}), this algorithm never sees the approval
profile: the oracle determines which candidate is exposed next. The theorem below shows that $L$ never contains more than $k$ candidates.
Its analysis uses the following lemma.

\begin{lemma}
\label{lem:ordered-witness-group-packing}
Let $S_1,\ldots,S_T\subseteq N$, and let
$\ell_1,\ldots,\ell_T\in[k]$. Suppose that
$\card{S_t}\ge \ell_t n/k$ for every $t$, and that $S_t$ intersects at
most $\ell_t-1$ of the sets $S_1,\ldots,S_{t-1}$. Then $T\le k$.
\end{lemma}

\begin{proof}
Starting with the entire family, repeatedly select a remaining set
$S_t$ of largest index and delete it together with every remaining set
that intersects it. Let $I$ be the set of selected indices. When $S_t$
is selected, every other remaining set has smaller index. By assumption,
$S_t$ intersects at most $\ell_t-1$ of them, so the step deletes at
most $\ell_t$ sets, including $S_t$ itself. Moreover, the selected sets
are pairwise disjoint. Since the deletion steps partition the original
family,
$$
T
\le \sum_{t\in I}\ell_t
\le \frac{k}{n}\sum_{t\in I}\card{S_t}
\le k.
$$
\end{proof}

\begin{theorem}
\label{thm:pjr-plus-upper-bound}
For every $k\ge1$, the Retention Algorithm incurs at most $k$ rejections under
candidate-only PJR+ feedback.
\end{theorem}

\begin{proof}
Suppose that the first $T$ proposals $W_1,\ldots,W_T$ are rejected.
Let $c_t$ be the candidate returned in round $t$, and choose an
associated PJR+ violation $(\ell_t,c_t,S_t)$.

Fix $s<t$. Since the algorithm retains every returned candidate,
$c_s\in W_t$, whereas $c_t\notin W_t$. Thus, the returned candidates
are pairwise distinct. Moreover, if $S_s\cap S_t\ne\varnothing$, then
a voter in the intersection approves $c_s$, and hence
$c_s\in W_t\cap\bigcup_{i\in S_t}A_i$. Since
$(\ell_t,c_t,S_t)$ is a PJR+ violation,
$\card{W_t\cap\bigcup_{i\in S_t}A_i}<\ell_t$. The distinctness of
the returned candidates therefore implies that
$S_t$ intersects at most $\ell_t-1$ earlier witness groups. Applying
\cref{lem:ordered-witness-group-packing} gives $T\le k$.
\end{proof}

\begin{theorem}
\label{thm:pjr-plus-lower-bound}
For every $k\ge1$, $Q_{\mathrm{full}}^{\PJR}(k)\ge k$.
\end{theorem}

\begin{proof}
Fix a deterministic full-witness learner, let $n=k$, and take $m>k$.
For each $t\in[k]$, after the learner proposes $W_t$, choose any
$c_t\in C\setminus W_t$ and respond with $(k,c_t,N)$.

After round $k$, choose a set $K\in\binom{C}{k}$ containing all the
returned candidates $c_1,\ldots,c_k$, and let every voter approve
exactly $K$. Such a set exists because at most $k$ distinct candidates
were returned. For every $t\in[k]$, we have $c_t\in K\setminus W_t$,
and therefore
$$
\card{W_t\cap\bigcup_{i\in N}A_i}
=\card{W_t\cap K}
<k.
$$
Since $\card{N}=k$ and every voter approves $c_t$, the triple
$(k,c_t,N)$ is a PJR+ violation of $W_t$ under this profile.

Because the learner is deterministic, these responses generate exactly
the proposals $W_1,\ldots,W_k$. Complete the oracle strategy by
accepting every PJR+ committee and returning an arbitrary PJR+ violation
otherwise. The resulting truthful full-witness oracle forces $k$
rejected proposals on this fixed profile.
\end{proof}

Together, \cref{thm:pjr-plus-upper-bound,thm:pjr-plus-lower-bound} show
that deterministic PJR+ rejection complexity is exactly $k$ under both
feedback models. The randomized lower bound proved later completes the
four-way characterization stated in the introduction.

The Retention Algorithm does not extend to EJR+. The next example shows
that, even after $k$ returned candidates fill the committee, another EJR+ violation may return a new candidate; permanent retention is then no
longer feasible.

\begin{example}
\label{ex:retention-fails-ejr-plus}
Let $n=k=2$, with $A_1=\{p_1,c\}$ and $A_2=\{p_2,c\}$, and let $z$ be
unapproved. Candidate-only retention may produce
\[
\{z,p_1\}\xrightarrow{p_2}\{z,p_2\}
\xrightarrow{p_1}\{p_1,p_2\}.
\]
Each response is a level-$1$ EJR+ violation for the voter approving the
returned candidate. The retained set now fills the committee, but
$(2,c,\{1,2\})$ is another violation because each voter approves only one
committee member, while both approve $c$. Retaining $c$ as well is impossible.
\end{example}

\section{Observable Potential Progress}
\label{sec:full-witness-ejr-plus-upper-bound}

The preceding example shows that a returned candidate cannot always be kept
forever under EJR+. Under full-witness feedback, potential progress is
observable on the revealed partial profile: a response certifies $c\in A_i$
for every $i\in S$. Recording these approvals produces a partial profile on
which the learner can certify PAV-improving swaps using only the information
revealed so far.

\begin{definition}[PAV score~\cite{LacknerSkowron2023}]
\label{def:pav-score}
Let $H_0:=0$ and $H_r:=\sum_{j=1}^r 1/j$ for $r\ge1$. For an approval
profile $A=(A_i)_{i\in N}$ and a feasible committee $W$, its
\emph{proportional approval voting (PAV) score} is
$$
\Phi_A(W):=\sum_{i\in N}H_{\card{A_i\cap W}}.
$$
\end{definition}

A voter's first approved committee member contributes $1$, the second
contributes $1/2$, and so on. The score therefore rewards improvements for
poorly represented voters more strongly. For $c\in C\setminus W$ and
$d\in W$, write
$W-d+c:=(W\setminus\{d\})\cup\{c\}$.

\citet{AzizEtAl2018} show that a committee admitting no swap of
PAV gain at least $n/k^2$ satisfies EJR, and \citet[Section~3.1]{CaseyElkind2025} observe that the same argument 
implies the same for EJR+. By contraposition, an EJR+ violation guarantees the existence of
an improving swap. The next lemma proves the stronger statement needed here:
for every EJR+ violation $(\ell,c,S)$, the average over all $k$ swaps that
insert that same returned candidate $c$ is at least $n/k^2$.

\begin{lemma}
\label{lem:ejr-plus-average-swap-gain}
Let $A$ be an approval profile, let $W$ be a feasible committee, and let
$(\ell,c,S)$ be an EJR+ violation of $W$. Then
$$
\frac{1}{k}\sum_{d\in W}
\bigl(\Phi_A(W-d+c)-\Phi_A(W)\bigr)
\ge \frac{n}{k^2}.
$$
\end{lemma}

\begin{proof}
For each voter $i\in N$, let $t_i:=\card{A_i\cap W}$, and let
$\delta_i(d)$ denote the change in voter $i$'s PAV score when $d$ is
replaced by $c$. Suppose first that $c\in A_i$. If $d\notin A_i$, then
the number of approved committee members increases from $t_i$ to
$t_i+1$, so $\delta_i(d)=1/(t_i+1)$. If $d\in A_i$, then the voter
loses one approved candidate and gains another, so $\delta_i(d)=0$.
Since exactly $k-t_i$ candidates in $W$ are not approved by $i$,
$$
\sum_{d\in W}\delta_i(d)=\frac{k-t_i}{t_i+1}.
$$
For $i\in S$, the violation gives $t_i<\ell$, and hence this quantity is
at least $(k-\ell+1)/\ell$.

If $c\notin A_i$, swapping out an unapproved candidate has no effect,
whereas swapping out any of the $t_i$ approved candidates decreases the
score by $1/t_i$. Thus
$\sum_{d\in W}\delta_i(d)=0$ when $t_i=0$ and equals $-1$ otherwise.
Consequently, every voter in $S$ contributes at least
$(k-\ell+1)/\ell$ to the sum over deletions, while every voter outside
$S$ contributes at least $-1$. Using $\card{S}\ge\ell n/k$, we obtain
\[
\begin{aligned}
&\sum_{d\in W}\bigl(\Phi_A(W-d+c)-\Phi_A(W)\bigr)=\sum_{i\in N}\sum_{d\in W}\delta_i(d)\\
& \quad \ge \card{S}\frac{k-\ell+1}{\ell}
-\bigl(n-\card{S}\bigr)\\
&\quad=
\frac{\card{S}(k+1)-n\ell}{\ell}
\ge \frac{n}{k}.
\end{aligned}
\]
Dividing by $k$ proves the claim.
\end{proof}

\begin{algorithmblock}{Witness-Guided PAV Local Search}
Fix an ordering of all swaps. Maintain a \emph{revealed partial profile}
$B=(B_i)_{i\in N}$, where $B_i$ contains the approvals of voter $i$
certified so far, initially $B_i=\varnothing$ for every $i\in N$, together
with a feasible committee $W$. While some
swap satisfies $\Phi_B(W-d+c)-\Phi_B(W)\ge n/k^2$, perform the first such
swap. When none exists, propose $W$. If the oracle returns $(\ell,c,S)$, add
$c$ to $B_i$ for every $i\in S$ and resume the local search from the same
committee.
\end{algorithmblock}
The fixed ordering only makes the learner deterministic. Since $c\notin W$,
recording the newly certified approvals leaves the current score unchanged;
as shown below, it also makes the returned triple an EJR+ violation of
the updated revealed partial profile.

\begin{theorem}
\label{thm:full-witness-ejr-plus-upper-bound}
For every $k\ge1$, Witness-Guided PAV Local Search incurs at most
$k^2H_k$ rejections and runs in polynomial time.
\end{theorem}

\begin{proof}
The invariant $B_i\subseteq A_i$ holds for every voter $i$, since the
learner records only approvals certified by the oracle.

Suppose that $W$ is rejected with violation $(\ell,c,S)$, and let $B'$ be
the revealed partial profile obtained by adding $c$ to $B_i$ for every $i\in S$.
Since $c\notin W$, this update does not change any voter's representation
in $W$: $\card{B'_i\cap W}=\card{B_i\cap W}$ for every $i\in N$, and hence
$\Phi_{B'}(W)=\Phi_B(W)$.
For every $i\in S$, we have $c\in B'_i$ and
$$
\card{B'_i\cap W}
=
\card{B_i\cap W}
\le
\card{A_i\cap W}
<
\ell.
$$
Together with $\card{S}\ge\ell n/k$, this shows that $(\ell,c,S)$ is an
EJR+ violation of $W$ under $B'$. By
\cref{lem:ejr-plus-average-swap-gain}, some $d\in W$ satisfies
$\Phi_{B'}(W-d+c)-\Phi_{B'}(W)\ge n/k^2$. Thus every rejection forces a
swap of gain at least $n/k^2$ before the next proposal.

The score $\Phi_B$ is initially zero, is unchanged by updates to the
revealed partial profile, and increases by at least $n/k^2$ at every
performed swap. Since it
is always at most $nH_k$, if $R$ swaps are performed, then
$$
R\frac{n}{k^2}\le nH_k.
$$
Hence $R\le k^2H_k$. Every rejection forces at least one subsequent swap,
so the learner incurs at most $k^2H_k$ rejections.

Each local-search step examines at most $k(m-k)$ swaps, each evaluable in
$O(n)$ time. Since at most $k^2H_k$ swaps occur, the running time is
polynomial in $n$, $m$, and $k$.
\end{proof}

\subsection{A Superlinear Lower Bound}
\label{sec:level-sequence-lower-bound}

The preceding upper bound leaves open whether EJR+ might still admit the
linear rejection complexity of PJR+. The next construction rules this out: even with full-witness feedback,
deterministic EJR+ learning can require a superlinear number of rejections.
By \cref{eq:complexity-monotonicity}, the same lower bound holds under
candidate-only feedback.

We encode each returned candidate by the number of voters who approve it,
which we call its \emph{support size}. A numerical condition will ensure
that any sequence of support sizes
can be forced against a deterministic learner. Long sequences will then
give strong lower bounds.

For this construction, set $n=k$, so the quota $n/k$ is one. Although a
level-$\ell$ violation may contain more than $\ell$ voters, our adversary
will choose a witness group of exactly $\ell$ voters. Thus the same integer
$\ell$ records both the level of the violation and the support size of its
returned candidate.

Suppose that the candidates returned in earlier rounds have support sizes
$\lambda_1,\lambda_2,\ldots$. The final profile will leave every other
candidate unapproved, and each newly returned candidate will lie outside all
earlier proposals. If the candidate returned in round $s$
belongs to a committee $W$, it contributes exactly $\lambda_s$ to
$\sum_{i\in N}\card{A_i\cap W}$. A committee contains at most $k$
previously returned candidates, so this sum is at most the sum of the $k$
largest preceding support sizes. The following counting lemma determines
when enough voters remain underrepresented for another violation.

\begin{lemma}
\label{lem:low-satisfaction-count}
Let $d_1,\ldots,d_k$ be nonnegative integers and let $\ell\in[k]$. If
$\sum_{i=1}^k d_i\le\ell(k-\ell+1)-1$, then at least $\ell$ indices
$i$ satisfy $d_i<\ell$.
\end{lemma}

\begin{proof}
If fewer than $\ell$ indices satisfied $d_i<\ell$, then at least
$k-\ell+1$ indices would satisfy $d_i\ge\ell$, and hence
$\sum_{i=1}^k d_i\ge\ell(k-\ell+1)$.
\end{proof}

\begin{definition}
\label{def:admissible-level-sequence}
A nondecreasing sequence $\lambda_1,\ldots,\lambda_T\in[k]$ is
\emph{admissible} if, for every $t\in[T]$,
$$
\sum_{s=\max\{1,t-k\}}^{t-1}\lambda_s
\le \lambda_t(k-\lambda_t+1)-1.
$$
\end{definition}

The length of any admissible sequence lower-bounds the number of
rejected proposals.

\begin{theorem}
\label{thm:admissible-sequence-full-witness-lower-bound}
Every admissible sequence of length $T$ satisfies
$Q_{\mathrm{full}}^{\EJR}(k)\ge T$.
\end{theorem}

\begin{proof}
Fix an admissible sequence $\lambda_1,\ldots,\lambda_T$ and an
arbitrary deterministic learner. Let $N=[k]$ and let $C$ contain
$(k+1)T$ candidates. After each proposal, the adversary chooses a
candidate $c_t$ and a voter set $S_t$ and returns
$(\lambda_t,c_t,S_t)$. We define these choices adaptively and then verify
them under one fixed approval profile.

Suppose that the first $t-1$ responses have been chosen. The candidates
$c_1,\ldots,c_{t-1}$ are distinct, and
$\card{S_s}=\lambda_s$ for every $s<t$. Let
$A_i^{t-1}:=\{c_s:s<t\text{ and }i\in S_s\}$ be the approvals assigned
to voter $i$ so far. Each candidate $c_s$ was chosen outside every
committee proposed up to round $s$.

Let $W_t$ be the learner's proposal in round $t$, and put
$d_i:=\card{A_i^{t-1}\cap W_t}$. Every previously returned candidate
$c_s\in W_t$ is approved by exactly $\lambda_s$ voters and therefore
contributes $\lambda_s$ to $\sum_i d_i$. Since $W_t$ contains at most
$k$ such candidates and the sequence is nondecreasing,
$$
\sum_{i\in N}d_i
=\sum_{\substack{s<t\\ c_s\in W_t}}\lambda_s
\le \sum_{s=\max\{1,t-k\}}^{t-1}\lambda_s
\le \lambda_t(k-\lambda_t+1)-1.
$$
The counting lemma therefore supplies at least $\lambda_t$ voters with
$d_i<\lambda_t$; the adversary chooses any $\lambda_t$ of them as
$S_t$.

The set
$W_1\cup\cdots\cup W_t\cup\{c_1,\ldots,c_{t-1}\}$ contains at most
$kt+t-1\le(k+1)T-1$ candidates. Since $\card{C}=(k+1)T$, some
candidate $c_t$ lies outside this set. The adversary declares that
$c_t$ is approved precisely by the voters in $S_t$ and returns
$(\lambda_t,c_t,S_t)$.

After round $T$, define $A_i:=\{c_s:i\in S_s\}$ for every voter $i$,
and let every other candidate be approved by no voter. Fix a round
$s$. Since $c_s$ was chosen outside $W_s$, we have $c_s\notin W_s$;
by construction, every voter in $S_s$ approves $c_s$. Since every
candidate $c_r$ introduced in a round $r\ge s$ was also chosen outside
$W_s$, for every $i\in S_s$ we have
$
\card{A_i\cap W_s}
=\card{A_i^{s-1}\cap W_s}
<\lambda_s.
$
Finally, $n=k$ and $\card{S_s}=\lambda_s$, so
$\card{S_s}=\lambda_s n/k$. Hence
$(\lambda_s,c_s,S_s)$ is an EJR+ violation of $W_s$ under $A$.

Because the learner is deterministic, returning these same responses
under $A$ causes it to make the same proposals $W_1,\ldots,W_T$.
Define the oracle on every other interaction history by accepting an
EJR+ committee and otherwise returning any EJR+ violation. This gives
one truthful full-witness oracle under which the learner incurs at
least $T$ rejected proposals, as desired.
\end{proof}

\subsubsection{An Explicit $\frac23 k^{3/2}$ Lower Bound}
\label{subsec:explicit-admissible-sequence}
\mbox{}\par
A simple block construction already gives a
$\Theta(k^{3/2})$ lower bound. Suppose we are adding copies of level $\ell$.
Before the last copy, at most $q_\ell-1$ entries in the preceding $k$-window
have level $\ell$, while every other entry has level at most $\ell-1$.
Setting the resulting worst-case sum at the admissibility threshold suggests
$$
q_\ell:=\max\{0,k-\ell(\ell-1)\}.
$$
We form a nondecreasing sequence containing $q_\ell$ copies of each level
$\ell\in[k]$.

\begin{theorem}
\label{thm:explicit-admissible-sequence-lower-bound}
The sequence above is admissible. Its length is
$$
B_k:=\sum_{\ell=1}^k q_\ell
=\frac23k^{3/2}+O(k).
$$
\end{theorem}

\begin{proof}
Consider any copy of a level $\ell$ with $q_\ell>0$. Among the $k$
largest preceding levels, with zeros added if necessary, at most
$q_\ell-1$ equal $\ell$, while every remaining entry is at most
$\ell-1$. Their sum is therefore at most
\[
\begin{aligned}
(q_\ell-1)\ell+(k-q_\ell+1)(\ell-1)
&=k(\ell-1)+q_\ell-1\\
&=\ell(k-\ell+1)-1.
\end{aligned}
\]
Thus every copy is admissible.

Let $r$ be the largest integer satisfying $r(r-1)<k$. Then
$q_\ell>0$ exactly for $\ell\le r$, and hence
$$
B_k
=rk-\sum_{\ell=1}^r\ell(\ell-1)
=rk-\frac{r^3-r}{3}.
$$
Since $r=\sqrt{k}+O(1)$, this equals
$\frac23k^{3/2}+O(k)$. By
\cref{thm:admissible-sequence-full-witness-lower-bound},
$Q_{\mathrm{full}}^{\EJR}(k)\ge B_k$, and hence
$Q_{\mathrm{full}}^{\EJR}(k)\ge\frac23k^{3/2}-O(k)$.
The candidate-only bound follows because candidate-only feedback is no
more informative than full-witness feedback.
\end{proof}

\begin{corollary}
\label{cor:asymptotic-full-witness-ejr-plus-lower-bound}
As $k\to\infty$,
$$
Q_{\mathrm{full}}^{\EJR}(k)
\ge \frac23k^{3/2}-O(k).
$$
The same bound holds for $Q_{\mathrm{cand}}^{\EJR}(k)$.
\end{corollary}

\begin{proof}
Apply \cref{thm:admissible-sequence-full-witness-lower-bound} to the
sequence in \cref{thm:explicit-admissible-sequence-lower-bound}. The
candidate-only statement follows from \cref{eq:complexity-monotonicity}.
\end{proof}

Appendix~\ref{sec:optimal-admissible-sequences} optimizes the
admissible-sequence construction and gives a refined asymptotic lower bound.

\section{Average Potential Progress}
\label{sec:candidate-only-ejr-plus}

Candidate-only feedback hides the witness group needed to compare deletions,
but the average-swap guarantee still holds for the hidden PAV score. We exploit
it first by randomization and then by deterministic branching, while also
identifying the limits of following a single update path.

\subsection{Uniform Random Deletion}
\label{subsec:uniform-random-deletion}

Start from any committee. When candidate $c$ is returned against $W$, choose
$d$ uniformly from $W$ and propose $W-d+c$; the oracle responds before this
fresh random choice. By \cref{lem:ejr-plus-average-swap-gain}, the hidden PAV
score then increases in expectation.

\begin{theorem}
\label{thm:randomized-candidate-only-ejr-plus-upper-bound}
For every $k\ge1$ and every hidden profile and adaptive oracle,
uniform random deletion terminates with probability one and incurs at most
$k^2H_k$ rejections in expectation.
\end{theorem}

\begin{proof}
Fix the hidden profile $A$, write $\Phi:=\Phi_A$, and let $\tau$ be the
number of rejections.  Freeze the process after acceptance by setting
$W_t=W_\tau$ for $t\ge\tau$. Conditional on the history through the response
$c_t$ in a rejected round, the deletion is uniform in $W_t$; hence
\cref{lem:ejr-plus-average-swap-gain} gives
\[
\E[\Phi(W_{t+1})-\Phi(W_t)\mid\text{history through }c_t]
\ge \frac{n}{k^2}.
\]
Therefore, for every $T\ge1$,
\[
\E[\Phi(W_{\tau\wedge T})]-\Phi(W_0)
\ge \frac{n}{k^2}\sum_{t=0}^{T-1}\Prb[t<\tau]
=\frac{n}{k^2}\E[\tau\wedge T].
\]
Since $0\le\Phi(W)\le nH_k$, we obtain
$\E[\tau\wedge T]\le k^2H_k$. Letting $T\to\infty$ gives
$\E[\tau]\le k^2H_k$ and, in particular, almost-sure termination.
\end{proof}

\subsection{A Randomized Lower Bound}
\label{subsec:randomized-full-witness-linear-lower-bound}

Randomization still cannot reduce the complexity below $k$, even with full
witnesses. The construction hides one universally approved candidate in each
of $k$ large blocks.

\begin{theorem}
\label{thm:randomized-full-witness-linear-lower-bound}
For every $k\ge1$ and $\mathsf X\in\{\PJR,\EJR\}$,
$\widetilde Q_{\mathrm{full}}^{\mathsf X}(k)\ge k$.
\end{theorem}

\begin{proof}
Fix an integer $M>k$, let $n=k$, and partition the candidate set into
pairwise disjoint blocks
$$
C=C_1\mathbin{\dot\cup}\cdots\mathbin{\dot\cup}C_k,
\qquad
\card{C_t}=M
\quad\text{for every }t\in[k].
$$
Fix an arbitrary randomized full-witness learner $\mathcal L$ for these
parameters.

Choose candidates $x_t\in C_t$ independently and uniformly, and set
$K:=\{x_1,\ldots,x_k\}$. For each realization
$x=(x_1,\ldots,x_k)$, define the approval profile $A^x$ by
$A_i^x:=K$ for every $i\in N$. Thus, every voter approves exactly the
same hidden committee $K$.

Fix an ordering of $C$. For each realization $x$, define an oracle
$\mathcal O_x$ as follows. In round $t\le k$, against a committee $W$,
the oracle returns $(k,x_t,N)$ if $x_t\notin W$. If $x_t\in W$, it
accepts when $W=K$ and otherwise returns $(k,c,N)$, where $c$ is the
first candidate in $K\setminus W$. After round $k$, it accepts $K$ and,
against every other committee $W$, returns $(k,c,N)$ for the first
candidate $c\in K\setminus W$.

The oracle is truthful for the fixed profile $A^x$. Indeed, if
$W\ne K$, then $K\setminus W\ne\varnothing$, and every
$c\in K\setminus W$ is approved by all voters, while
$\card{A_i^x\cap W}=\card{K\cap W}<k$ for every $i\in N$. Since $n=k$, the triple $(k,c,N)$ is an EJR+ violation of $W$. It is
also a PJR+ violation, because
$\card{W\cap\bigcup_{i\in N}A_i^x}=\card{W\cap K}<k$. The committee $K$
satisfies both axioms, because no candidate outside $K$ is approved by any
voter.

For each $t\in[k]$, let $\widehat W_t$ be the random committee that
$\mathcal L$ would propose after receiving the prescribed responses
$$
(k,x_1,N),\ldots,(k,x_{t-1},N).
$$
The committee $\widehat W_t$ depends on $x_1,\ldots,x_{t-1}$ and the
learner's randomness, but not on $x_t$. Define $E_0$ to be the certain
event and, for $t\in[k]$, let
$$
E_t:=\bigcap_{s=1}^t\{x_s\notin\widehat W_s\}.
$$
On $E_t$, an induction on $s$ shows that the actual first $t$ proposals
are $\widehat W_1,\ldots,\widehat W_t$ and the oracle returns
$(k,x_s,N)$ in every round $s\le t$. In particular, the first $t$
proposals are rejected.

The event $E_{t-1}$ and the committee $\widehat W_t$ are determined
without reference to $x_t$. Since $x_t$ is uniform on $C_t$ and
independent of all preceding hidden candidates and of the learner's
randomness,
$$
\begin{aligned}
\Prb(E_t\mid E_{t-1})
&=1-\E\left[
\frac{\card{\widehat W_t\cap C_t}}{M}
\,\middle|\,
E_{t-1}
\right]\\
&\ge 1-\frac{k}{M}.
\end{aligned}
$$
The events are nested, and therefore
$$
\Prb(E_k)
=\prod_{t=1}^k\Prb(E_t\mid E_{t-1})
\ge\left(1-\frac{k}{M}\right)^k.
$$

Let $\tau$ denote the number of rejected proposals. On $E_k$, the first
$k$ proposals are rejected, so
$$
\E_{x,\mathcal L}\left[
\tau(\mathcal L,A^x,\mathcal O_x)
\right]
\ge k\left(1-\frac{k}{M}\right)^k.
$$
The expectation on the left is the average, over all realizations $x$,
of the learner's expected number of rejections against the fixed pair
$(A^x,\mathcal O_x)$. Hence some realization $x^\star$ satisfies
$$
\E_{\mathcal L}\left[
\tau(\mathcal L,A^{x^\star},\mathcal O_{x^\star})
\right]
\ge k\left(1-\frac{k}{M}\right)^k.
$$

Since $\mathcal L$ was arbitrary, the randomized rejection complexity
for either target axiom, with $n=k$ and $m=kM$, is at least the right-hand
side. Taking the supremum over $M>k$ gives, for
$\mathsf X\in\{\PJR,\EJR\}$,
$$
\widetilde Q_{\mathrm{full}}^{\mathsf X}(k)
\ge
\sup_{M>k}k\left(1-\frac{k}{M}\right)^k
=k.
$$
\end{proof}

\subsection{The Single-Path Barrier}
\label{subsec:single-swap-lower-bound}

Fix an initial committee $W_0$. An \emph{insertion--deletion learner} always
responds to a returned candidate $c$ by proposing $W-d+c$ for some
$d\in W$; the deletion may be randomized.

\begin{theorem}
\label{thm:insertion-deletion-ejr-plus-lower-bound}
For sufficiently many candidates, there are a fixed profile and a
candidate-only oracle that force every insertion--deletion learner from $W_0$
to incur at least $B_k$ rejections, where
$B_k=\sum_{\ell=1}^k q_\ell=\frac23k^{3/2}+O(k)$. In particular, the
number of rejections is at least $\frac23k^{3/2}-O(k)$.
\end{theorem}

\begin{proof}
Let $\lambda_1,\ldots,\lambda_T$ be the explicit admissible sequence from
\cref{thm:explicit-admissible-sequence-lower-bound}, so $T=B_k$.
Build a complete rooted $k$-ary tree of depth $T$ and label its root by
$W_0$. Consider a node $v$ at depth $t<T$. Its committee $W_v$ consists
only of root candidates and candidates returned at ancestors of $v$.
For each voter $i$, let $d_i(v)$ be the number of ancestor candidates in
$W_v$ approved by $i$. Since $W_v$ contains at most $k$ ancestor
candidates, whose support sizes are drawn from
$\lambda_1,\ldots,\lambda_t$, admissibility gives
$$
\sum_{i=1}^k d_i(v)
\le\sum_{s=\max\{1,t-k+1\}}^t\lambda_s
\le\lambda_{t+1}(k-\lambda_{t+1}+1)-1.
$$
By \cref{lem:low-satisfaction-count}, choose a set $S_v$ of
$\lambda_{t+1}$ voters with $d_i(v)<\lambda_{t+1}$. Introduce a
candidate $c_v$, unique to $v$, approved precisely by $S_v$; call $c_v$ the
node candidate at $v$. For each
$d\in W_v$, create a child labelled $W_v-d+c_v$.

Define one fixed profile in which every node candidate $c_v$ is approved
precisely by $S_v$, and let the root candidates and all remaining candidates
be approved by no voter.
At node $v$, the committee contains only candidates associated with
ancestors of $v$; it contains no descendant candidate and no candidate
from another branch. Hence each $i\in S_v$ approves exactly $d_i(v)$
members of $W_v$, while all voters in $S_v$ approve
$c_v\notin W_v$. Since $\card{S_v}=\lambda_{t+1}=\lambda_{t+1}n/k$,
$(\lambda_{t+1},c_v,S_v)$ is an EJR+ violation of $W_v$.

The oracle uses the realized history to identify the current node and
returns its designated candidate. Whatever committee member the learner
deletes, the next committee is the corresponding child. Every
root-to-leaf interaction therefore contains $T$ rejected proposals.
The stated bound on $m$ supplies the $k$ root candidates and one distinct
candidate for every internal node. On any other history, the oracle accepts
an EJR+ committee and otherwise returns the candidate component of an
arbitrary EJR+ violation. Thus the oracle is defined on every history.
\end{proof}

\subsection{Exhaustive Branching}
\label{subsec:deterministic-candidate-only-branching}

An unrestricted deterministic learner can avoid committing to one deletion
by exploring every branch. Because the PAV score is hidden, it cannot select a
good child by conditional expectations \citep[Chapter~5]{MotwaniRaghavan1995}.
The next lemma shows that average potential growth nevertheless forces some
branch to terminate.

\begin{lemma}
\label{lem:average-drift-branching}
Let $\mathcal X$ be a proposal space. Suppose that a rejected proposal
$x$ and response $r$ determine successors
$T_1(x,r),\ldots,T_b(x,r)$, and that a hidden potential
$\Psi:\mathcal X\to[0,B]$ satisfies
$$
\frac1b\sum_{j=1}^b\Psi(T_j(x,r))\ge\Psi(x)+\delta
$$
for every rejected proposal and response. Then querying every
node of the resulting proposal tree through depth
$\lfloor B/\delta\rfloor+1$ must encounter an accepted proposal, even
when the responses are chosen adaptively.
\end{lemma}

\begin{proof}
Suppose every node through depth $L:=\lfloor B/\delta\rfloor+1$ is rejected.
After the adaptive response at each node is fixed, choose a random path by
selecting one child uniformly at every step. Conditional on the current node,
the expected potential increases by at least $\delta$, so
$\E[\Psi(X_L)]\ge\Psi(X_0)+L\delta>B$, contradicting $\Psi\le B$.
\end{proof}

\begin{theorem}
\label{thm:deterministic-candidate-only-ejr-plus-upper-bound}
For every $k\ge1$,
$Q_{\mathrm{cand}}^{\EJR}(k)\le 2^{O(k^2(\log k)^2)}$.
\end{theorem}

\begin{proof}
For $k\ge2$, set $L:=\lfloor k^2H_k\rfloor+1$. Query an arbitrary root
committee. Whenever $W$ is rejected and $c$ is returned, create the $k$
children $W-d+c$, one for each $d\in W$, and query the resulting tree
through depth $L$ in breadth-first order, stopping at the first accepted node.

Apply \cref{lem:average-drift-branching} to the PAV score of the hidden profile,
with $b=k$, $B=nH_k$, and $\delta=n/k^2$. Since
$L=\lfloor B/\delta\rfloor+1$, some node through depth $L$ is accepted.
The tree has at most
$$
\sum_{j=0}^{L}k^j
\le k^{L+1}
$$
nodes. Since $H_k=O(\log k)$, we have
$L=O(k^2\log k)$, and hence the number of rejected proposals is
$2^{O(k^2(\log k)^2)}$.

If $k=1$, any returned candidate is approved by every voter, so the
singleton committee containing it is accepted after at most one
rejection.
\end{proof}

\section{Discussion}
\label{sec:discussion}

Taken together, our results reveal a hierarchy in the progress available from
violation feedback: permanent retention for PJR+, observable potential progress
for EJR+ with full witnesses, and only average potential progress under
candidate-only feedback.

Several quantitative questions remain open. With full witnesses, deterministic
EJR+ complexity lies between $\Omega(k^{3/2})$ and $O(k^2\log k)$;
with candidate-only feedback, the same lower bound contrasts with an
exponential upper bound, and it is unknown whether a polynomial bound exists.
Likewise, the single-path lower bound does not apply to learners that retain
multiple deletion branches. For EJR+, we conjecture that even a randomized
learner can be forced to make more than a linear number of proposals.

\begin{conjecture}
\label{conj:randomized-ejr-superlinear}
$\widetilde Q_{\mathrm{full}}^{\EJR}(k)=\omega(k)$.
\end{conjecture}

\section*{Disclosure of AI Use}
All proofs were generated and written by the human author. The author used an LLM for editorial assistance. The author takes full responsibility for the mathematical content and the final text.

\bibliographystyle{plainnat}
\bibliography{references}

\appendix

\section{Optimal admissible sequences and the refined lower-bound constant}
\label{sec:optimal-admissible-sequences}

This appendix sharpens the construction based on candidate support sizes from
\cref{sec:level-sequence-lower-bound} in two ways. A greedy recurrence
computes the maximum possible length $L_k$ of an admissible sequence. A
continuum analysis then gives
$L_k\ge(\gamma-o(1))k^{3/2}$ for an explicit constant
$\gamma\approx1.1496$.

\subsection{The optimal greedy recurrence}
\label{subsec:greedy-level-recurrence}

By \cref{thm:admissible-sequence-full-witness-lower-bound}, obtaining
the strongest lower bound from this construction amounts to finding
the longest admissible sequence. At any point, only the $k$ largest
preceding levels matter, because they determine the largest possible value of
$\sum_i\card{A_i\cap W}$ for a proposed committee $W$. We therefore
represent the preceding levels by a nondecreasing state
$x=(x_1,\ldots,x_k)$ containing their $k$ largest values, with zeros
added initially.

Every admissible next level yields one additional rejection. Choosing
a smaller level is preferable because it gives the same immediate
benefit while adding less support to the state. This leads to the following
recurrence.

\begin{definition}
\label{def:greedy-level-recurrence}
Let $x=(0,\ldots,0)$. A level $\ell\in[k]$ is admissible at the state
$x=(x_1,\ldots,x_k)$ if
$$
\ell\ge x_k
\qquad\text{and}\qquad
\sum_{j=1}^k x_j\le\ell(k-\ell+1)-1.
$$
If such a level exists, append the smallest one and replace $x$ by the
sorted vector of the $k$ largest entries of
$x_1,\ldots,x_k,\ell$. Let $L_k$ be the number of levels appended
before no admissible level remains.
\end{definition}

The state always consists of the $k$ largest levels appended so far.
The condition $\ell\ge x_k$ keeps the sequence nondecreasing, while
the second condition is exactly the admissibility inequality.

\begin{proposition}
\label{prop:greedy-level-recurrence-optimal}
The recurrence in \cref{def:greedy-level-recurrence} produces an
admissible sequence of maximum possible length.
\end{proposition}

\begin{proof}
For a state $x$ and a level $\ell$ admissible at $x$, let
$\sigma_\ell(x)$ denote the successor state obtained by appending
$\ell$. A fixed level $\ell$ can appear at most $k$ times: before a
$(k+1)$st copy, $\sum_{j=1}^k x_j\ge k\ell$, whereas
admissibility would require
$k\ell\le\ell(k-\ell+1)-1<k\ell$. Since the levels are nondecreasing
and belong to $[k]$, every admissible sequence has length at most
$k^2$; in particular, the recurrence terminates.

For states $x$ and $y$, write $x\le y$ if $x_j\le y_j$ for every
$j\in[k]$. Suppose $x\le y$ and a level $\ell$ is
admissible at $y$. Since $x_k\le y_k\le\ell$ and
$$
\sum_{j=1}^k x_j
\le \sum_{j=1}^k y_j
\le \ell(k-\ell+1)-1,
$$
the same level is admissible at $x$. Taking the $k$ largest entries
preserves componentwise order, so
$\sigma_\ell(x)\le\sigma_\ell(y)$. Repeating this argument after each appended level shows that every
sequence of subsequent levels admissible from $y$ is also admissible from $x$.

Now fix a state $x$, let $a$ be its smallest admissible level, and let
$b$ be any other admissible level. Since $a\le b$, we have
$\sigma_a(x)\le\sigma_b(x)$. The preceding monotonicity argument then shows that every
sequence of subsequent levels admissible after choosing $b$ is also admissible after
choosing $a$. Since either choice appends one level,
choosing $a$ permits a sequence at least as long as choosing $b$.
Applying this argument at every state proves the proposition.
\end{proof}

Combining
\cref{thm:admissible-sequence-full-witness-lower-bound,prop:greedy-level-recurrence-optimal}
gives
$$
Q_{\mathrm{full}}^{\EJR}(k)\ge L_k
\qquad\text{and}\qquad
Q_{\mathrm{cand}}^{\EJR}(k)\ge L_k.
$$

\subsection{A refined asymptotic lower bound}
\label{subsec:asymptotic-level-sequence-lower-bound}

The explicit sequence in
\cref{thm:explicit-admissible-sequence-lower-bound} gives the constant
$2/3$. The following scaling calculation motivates a longer sequence,
which we then verify rigorously. The relevant levels are of order
$\sqrt{k}$, and each such level may occur $\Theta(k)$ times. To estimate
the admissible number of copies,
consider a proposed next level $\ell$ and let
$a_1,\ldots,a_k$ be the $k$ largest preceding levels, with zeros added
if necessary. Define their total deficit from $\ell$ by
$D:=\sum_{j=1}^k(\ell-a_j)$. Since
$\sum_j a_j=k\ell-D$, the admissibility condition becomes
\begin{equation}
D\ge\ell(\ell-1)+1.
\label{eq:deficit-admissibility}
\end{equation}
For $\ell=y\sqrt{k}$, the required normalized deficit is therefore
asymptotically $y^2$.

Fix $x\in(0,1]$, and suppose that each recent level occurs
approximately $xk$ times. Let $p$ be the integer satisfying
$x\in(1/(p+1),1/p]$. The $k$ largest preceding levels then consist of
$p$ complete blocks at deficits $0,1,\ldots,p-1$, followed by the
remaining $(1-px)k$ entries at deficit $p$. Their normalized total
deficit is
$$
F(x):=x\sum_{r=0}^{p-1}r+p(1-px)
=p-\frac{p(p+1)}2x.
$$
The adjacent formulas agree at their endpoints, so
$F:(0,1]\to[0,\infty)$ is continuous and strictly decreasing.

The block fraction at the scaled level $y$ for which the normalized deficit
equals $y^2$ is the value $h(y)$ satisfying $F(h(y))=y^2$. Equivalently, whenever
$(p-1)/2\le y^2\le p/2$,
$$
h(y):=\frac{2(p-y^2)}{p(p+1)}.
$$
These branches meet at their endpoints. On every bounded interval, only
finitely many branches occur and their derivatives are bounded, so $h$ is
Lipschitz. A level near $y\sqrt{k}$ can therefore occur approximately
$kh(y)$ times. Since a scaled interval of width $dy$ contains
approximately $\sqrt{k}\,dy$ integer levels, the predicted total
sequence length is
$$
k^{3/2}\int_0^\infty h(y)\,dy.
$$
Define
$$
\gamma:=\int_0^\infty h(y)\,dy
=\sum_{p=1}^\infty\frac{2}{p(p+1)}
\left[p(b_p-a_p)-\frac{b_p^3-a_p^3}{3}\right],
$$
where $a_p:=\sqrt{(p-1)/2}$ and $b_p:=\sqrt{p/2}$. The summand is
$\Theta(p^{-3/2})$, so the series converges, and numerical evaluation gives
$\gamma=1.1496\ldots$. We use this exact definition of $\gamma$ below.

\begin{theorem}
\label{thm:asymptotic-greedy-level-sequence}
The greedy sequence satisfies
$$
\liminf_{k\to\infty}\frac{L_k}{k^{3/2}}\ge\gamma.
$$
\end{theorem}

\begin{proof}
Fix $\varepsilon\in(0,1/2)$ and $Y>0$. For every
$1\le\ell\le\lfloor Y\sqrt{k}\rfloor$, take
$$
q_{\ell,k}:=
\left\lfloor
(1-\varepsilon)k
h\left(\frac{\ell}{\sqrt{k}}\right)
\right\rfloor
$$
copies of level $\ell$, and concatenate these blocks in increasing
order. It remains to prove that this sequence is admissible for all sufficiently
large $k$.

Consider the last copy of level $\ell$, put
$y:=\ell/\sqrt{k}$, and let $D_{\ell,k}$ be the total deficit from
$\ell$ of the $k$ largest preceding levels. By
\cref{eq:deficit-admissibility}, it suffices to prove
$D_{\ell,k}\ge\ell(\ell-1)+1$.

Write $x_y:=(1-\varepsilon)h(y)$. We claim that, uniformly for
$\ell\le Y\sqrt{k}$,
\begin{equation}
\frac{D_{\ell,k}}k
\ge F(x_y)-O(k^{-1/2}).
\label{eq:continuum-deficit-approximation}
\end{equation}
Choose $\eta>0$ such that $x_y>1/2$ for $0\le y\le\eta$. Before the
last copy of a level $\ell\le\eta\sqrt{k}$, at most $q_{\ell,k}-1$
entries have deficit zero, while every remaining entry has deficit at
least one. Hence
$$
D_{\ell,k}\ge k-q_{\ell,k}
=k(1-x_y)+O(1)
=kF(x_y)+O(1),
$$
where the final equality uses the first branch of $F$.

Now suppose $\eta\sqrt{k}\le\ell\le Y\sqrt{k}$. Since $h$ is positive
and Lipschitz on $[0,Y]$, there is a constant $x_0>0$ such that
$x_y\ge x_0$, and for every fixed $r$,
$q_{\ell-r,k}=kx_y+O(\sqrt{k})$ uniformly in $\ell$. Choose a fixed
integer $R>2/x_0$. For sufficiently large $k$, the blocks at levels
$\ell,\ell-1,\ldots,\ell-R$ contain more than $k$ entries, so the
$k$ largest preceding levels all lie in these blocks. If their block
sizes were exactly $kx_y$, their normalized deficit would be
$F(x_y)$. The actual boundary of each block differs by
$O(\sqrt{k})$, and all relevant deficits are at most $R$; hence the
total deficit differs by at most $O(\sqrt{k})$. This proves
\cref{eq:continuum-deficit-approximation}.

Because $F$ is strictly decreasing and
$x_y=(1-\varepsilon)h(y)<h(y)$, the function
$y\mapsto F((1-\varepsilon)h(y))-F(h(y))$ is continuous and strictly
positive on $[0,Y]$. It therefore has a positive minimum $\delta$.
Since $F(h(y))=y^2$, \cref{eq:continuum-deficit-approximation} gives
$$
D_{\ell,k}
\ge k\left(y^2+\delta-O(k^{-1/2})\right)
\ge \ell^2+\frac{\delta k}{2}
$$
for all sufficiently large $k$, uniformly over
$\ell\le Y\sqrt{k}$. Since the final quantity exceeds
$\ell(\ell-1)+1$, the last copy of every block is admissible. Earlier
copies in the same block have fewer zero-deficit predecessors and
therefore only larger deficit, so the entire sequence is admissible.

The floor in the definition of $q_{\ell,k}$ contributes total error
$O(\sqrt{k})$. A Riemann-sum approximation then gives
$$
\begin{aligned}
\sum_{\ell\le Y\sqrt{k}}q_{\ell,k}
&=(1-\varepsilon)k
  \sum_{\ell\le Y\sqrt{k}}h(\ell/\sqrt{k})+O(\sqrt{k})\\
&=(1-\varepsilon)k^{3/2}
  \int_0^Y h(y)\,dy+o(k^{3/2}).
\end{aligned}
$$
Since $L_k$ is the maximum length of an admissible sequence,
$$
\liminf_{k\to\infty}\frac{L_k}{k^{3/2}}
\ge(1-\varepsilon)\int_0^Y h(y)\,dy.
$$
Letting $\varepsilon\downarrow0$ and then $Y\to\infty$ proves the
claim.
\end{proof}

\begin{corollary}
\label{cor:refined-asymptotic-full-witness-ejr-plus-lower-bound}
As $k\to\infty$,
$$
Q_{\mathrm{full}}^{\EJR}(k)
\ge(\gamma-o(1))k^{3/2}.
$$
This lower bound also applies to $Q_{\mathrm{cand}}^{\EJR}(k)$.
\end{corollary}

\begin{proof}
The full-witness bound follows from
\cref{thm:admissible-sequence-full-witness-lower-bound,prop:greedy-level-recurrence-optimal,thm:asymptotic-greedy-level-sequence}.
The candidate-only statement follows from \cref{eq:complexity-monotonicity}.
\end{proof}

\end{document}